\documentclass[
final
]{dmtcs-episciences}

\usepackage[utf8]{inputenc}

\usepackage{amssymb,amsmath,amsthm}
\usepackage{epsfig}
\usepackage{footnpag}
\usepackage{url}

\newtheorem{thm}{Theorem}
\newtheorem{cor}[thm]{Corollary}

\newtheorem{prop}[thm]{Proposition}

\newtheorem{obs}[thm]{Observation}

\theoremstyle{remark}
\newtheorem{remark}[thm]{Remark}

\usepackage{indentfirst}
\usepackage{xspace} 

\mathcode`l="8000
\begingroup
\makeatletter
\lccode`\~=`\l
\DeclareMathSymbol{\lsb@l}{\mathalpha}{letters}{`l}
\lowercase{\gdef~{\ifnum\the\mathgroup=\m@ne \ell \else \lsb@l \fi}}%
\endgroup

\def\HH{\mbox{\ensuremath{\mathcal H}}\xspace}

\author{D\"om\"ot\"or P\'alv\"olgyi\affiliationmark{1}\thanks{Research supported by the Marie Sk\l odowska-Curie action of the EU, under grant IF 660400 and by the Lend\"ulet program of the Hungarian Academy of Sciences (MTA), under grant number LP2017-19/2017.}}
\title{Weak embeddings of posets to the Boolean lattice}
\affiliation{	
	MTA-ELTE Lend\"ulet Combinatorial Geometry Research Group, Institute of Mathematics, E\"otv\"os Lor\'and University (ELTE), Budapest, Hungary}
  
	
	\received{2017-5-24}
	
	\revised{2017-11-28}
	
	\accepted{2018-1-13}
	
	\publicationdetails{20}{2018}{1}{7}{3684}

\begin{document}
\maketitle

\begin{abstract}
The goal of this paper is to prove that several variants of deciding whether a poset can be (weakly) embedded into a small Boolean lattice, or to a few consecutive levels of a Boolean lattice, are {\bf NP}-complete, answering a question of Griggs and of Patk\'os.
As an equivalent reformulation of one of these problems, we also derive that it is {\bf NP}-complete to decide whether a given graph can be embedded into the two middle levels of some hypercube.
\end{abstract}

\section{Introduction}

A {\em poset} $(P,\le)$ is a partially ordered set on $|P|$ elements.
An injective map $f$ from poset $P$ to poset $Q$ is called a {\em weak embedding} if for every $p,q\in P$ we have $f(p)\le f(q)$ if $p\le q$, and it is called a {\em (strong) embedding} if $f(p)\le f(q)$ if and only if $p\le q$.
Similarly, an injective map $f$ from graph $G$ to graph $H$ is called an {\em embedding} if for any edge $uv$ of $G$ its image $f(u)f(v)$ is an edge of $H$, and it is called an {\em induced embedding} if $uv$ is an edge of $G$ if and only if $f(u)f(v)$ is an edge of $H$.
(Be careful that simply {\em embedding} is {\em strong} embedding for posets, but for graphs, it is the equivalent of {\em weak} poset embeddings --- unfortunately, both are standard terminology.)
The {\em Hasse diagram} of a poset $P$ is a graph whose vertex set is $V(P)$ and whose edges are pairs $\{p,q\}$ satisfying that there is no $r$ for which $p<r<q$ holds.
When speaking of elements of $P$, the terms \emph{neighborhood} and \emph{distance} refer to the Hasse diagram of $P$. 
Thus, we call the elements adjacent to an element $x$ in the Hasse diagram the {\em neighbors} of $x$, and the length of the shortest path in the Hasse diagram connecting some elements $x$ and $y$ their {\em distance}.
The {\em Boolean lattice} of $\{1,\ldots,n\}$, $B_n$, has $2^n$ elements, one for each subset of $\{1,\ldots,n\}$, where the ordering is given by containment structure, i.e., $X\le Y$ if $X\subset Y$.
The {\em $k^\mathrm{th}$ level} of a Boolean lattice is the collection of its elements of size $k$.
For $n$ even, we refer to the $(\frac n2)^\mathrm{th}$ level of $B_n$ as its {\em middle level}, while for general $n$, we refer to the levels from $\lfloor\frac {n-e+1}2\rfloor^\mathrm{th}$ to $\lfloor\frac {n+e-1}2\rfloor^\mathrm{th}$ as the $e$ {\em middle levels} of $B_n$.\\

In this paper we study the decision complexity of whether a poset admits a weak embedding to (some levels of) $B_n$ (where $n$ is arbitrary, given as part of the input).
Apparently, earlier only strong embeddings to $B_n$ have been studied, first in \cite{Trotter}, while the {\bf NP}-completeness of the problem was established in \cite{Stahl}; for more recent results related to complexity, see \cite{Habib,Raynaud}.
We find it somewhat surprising that weak embeddings have not yet been studied.
There are, however, some graph problems that are equivalent to weak embedding questions to two consecutive levels, e.g., the Middle Levels conjecture is that there is a Hamiltonian cycle in the union of the two middle levels of every $B_{2n+1}$ --- this has been recently solved by M\"utze~\cite{Mutze}; for a shorter proof, see~\cite{GMN}.\\

We write $P\subset Q$ if $P$ has a weak embedding to $Q$.
This indeed defines a partial order on the posets, i.e., $P\subset Q \subset R$ implies $P\subset R$ and $P\subset Q\subset P$ implies that $P$ and $Q$ are isomorphic.
If $P\subset Q$, we say that $Q$ {\em contains} (a copy of) $P$, otherwise we say that $Q$ is {\em $P$-free}.
We denote by $d(P)$ the {\em smallest} integer such that $P\subset B_{d(P)}$.
(For strong embeddings, this parameter is called the {\em 2-dimension} of $P$, and embeddings to $B_n$ are called {\em bit-vector encodings}.)
As $P\subset C_{|P|}\subset B_{|P|-1}$, where $C_n$ denotes the {\em chain} (totally ordered poset) on $n$ elements, $d(P)$ is always some non-negative integer.
Despite the huge literature of embedding trees to the hypercube \cite{CL11,LS88}, it seems that $d(P)$ has not even been studied for trees.
The problem of determining the value of $d(T_k)$, where $T_k$ denotes the complete binary tree of depth $k$, can be shown to be equivalent to a search problem proposed by G.O.H.\ Katona~\cite{Emlektabla}, which is also open.\\

We also study weak embeddings to the union of a few consecutive levels of the Boolean lattice.
We denote by $e(P)$ the {\em largest} integer such that any $e(P)$ consecutive levels of any Boolean lattice are $P$-free.
It follows from the definitions that $e(P)\le d(P)$, as any $d(P)+1$ levels of any Boolean lattice contain a copy of $B_{d(P)}$ which contains a copy of $P$.
If $P$ has a smallest and a largest element, then $e(P)=d(P)$, while examples for small posets for which inequality holds include the so-called {\em Fork} poset on three elements, $a,b,c$, with $a<b,c$, for which $e(P_\mathrm{fork})=1<d(P_\mathrm{fork})=2$,
and the so-called {\em Butterfly} poset on four elements, $w,x,y,z$, with $w,x<y,z$, for which $e(P_\mathrm{butterfly})=2<d(P_\mathrm{butterfly})=3$.
We also note that $h(P)-1\le e(P)\le d(P)$, where $h(P)$ is the {\em height} of the poset, i.e., the {\em cardinality} of its longest subchain.\\

The parameter $e(P)$ has been introduced in Griggs, Li and Lu~\cite{GriggsLiLu}, as it naturally came up while studying the largest possible size of a $P$-free subposet of $B_n$, denoted by $La(n,P)$.
This parameter has been first studied by Katona in the 1980s for general posets; for a recent survey see Griggs and Li~\cite{GriggsLi}.
The general conjecture, implicitly contained in the earlier works of Katona and others, and explicitly first stated by Bukh~\cite{Bukh}, and a couple of months later, independently, by Griggs and Lu~\cite{GriggsLu}, is that $\pi(P)=\lim_{n\to\infty} \frac{La(n,P)}{\binom{n}{n/2}}$ always exists, and equals to $e(P)$.
(Note that $e(P)\le \pi(P)$ follows from that the union of the $e(P)$ middle levels of $B_n$ are $P$-free.) 
This has only been proved for special posets.
The most general result is due to Bukh~\cite{Bukh}, which says that if the Hasse diagram of $P$ is a tree, then $\pi(P)=h(P)-1=e(P)$.\\

Motivated by this, Griggs~\cite{Griggs} and Patk\'os~\cite{Patkos} asked independently around the same time the complexity of determining $e(P)$.\footnote{Griggs has also asked for the complexity of determining the 2-dimension of $P$, but this has already been proved to be {\bf NP}-complete by Stahl and Wille~\cite{Stahl}; for a more accessible version, see Habib et al.~\cite{Habib}.}
Answering their questions, we show the following.

\begin{thm}\label{thm:d} To decide whether $d(P)$ is at most $n$ and to decide whether $e(P)$ is at most $n$ are both {\bf NP}-complete.
\end{thm}

\begin{remark} In fact, as we will see from the proof, it is already {\bf NP}-complete for posets with a smallest and a largest element (in which case $d(P)=e(P)$) to determine whether these parameters equal $h(P)-1$.
\end{remark}

\begin{thm}\label{thm:e} To decide whether $e(P)\le 1$ is {\bf NP}-complete.
\end{thm}

\begin{remark} The graph theoretic reformulation of Theorem~\ref{thm:e} is that it is {\bf NP}-complete to decide whether a given graph can be embedded into two consecutive levels of some hypercube.
\end{remark}

\begin{thm}\label{thm:34} To decide whether a poset can be weakly embedded into the union of the third and fourth level of some Boolean lattice is {\bf NP}-complete.
\end{thm}

\begin{remark} Both Theorems~\ref{thm:e} and \ref{thm:34} also hold for strong embeddings, as the respective posets used in their proofs can only have a strong embedding to the required structures (see Corollary \ref{weak2strong}).
\end{remark}

Finally, using our methods we also sketch the proof of a related result.

\begin{thm}\label{thm:j} To decide whether a graph is an induced subgraph of a Johnson graph is {\bf NP}-complete.
\end{thm}


\section{Preliminaries}

\subsection{Connection to graph embeddings}

It is well-known that directed and undirected graph embedding problems can be easily reduced to each other by simple gadgets.
\footnote{The interested reader can find a collection of similar reductions in Booth and Colbourn~\cite{BoothColbourn}.}
The same is true for weak poset embedding problems.
To reduce a weak poset embedding problem to a directed graph embedding problem, notice that $P$ weakly embeds to $Q$ if and only if the {\em transitive closure} of $P$ embeds to the {\em transitive closure} of $Q$.
To reduce a graph embedding problem to a weak poset embedding problem, let us denote by $\hat G$ the two-level poset obtained from a graph $G$ as follows.
The elements of $\hat G$ are the vertices and edges of $G$, and any edge is larger than its endpoints (these are the only relations).
Thus, the vertices of $G$ form an antichain in $\hat G$, the lower level, and the edges of $G$ also form an antichain in $\hat G$, the upper level.

\begin{prop}\label{prop:graph2poset} $G$ is a subgraph of $H$ if and only if $\hat G$ weakly embeds to $\hat H$.
\end{prop}

The interested reader can find the simple proof of Proposition \ref{prop:graph2poset} in \cite{BoothColbourn}.
As deciding whether a graph is a subgraph of another graph, known as the {\sc subgraph isomorphism} problem, is {\bf NP}-complete \cite{Cook}, we get that weak embedding for posets is also {\bf NP}-complete.

\begin{cor} Deciding whether $P$ weakly embeds to $Q$ or not is {\bf NP}-complete, already if both $P$ and $Q$ have only two levels.
\end{cor}

\begin{remark} Note that Theorem~\ref{thm:d} is {\em not} a strengthening of this corollary, as here $Q$ is also given as part of the input, while in Theorem \ref{thm:d} $B_n$ has exponential size (but its description, the binary encoding of $n$, is $\log \log$ of the size of $B_n$).
\end{remark}


\subsection{Uniqueness of embedding of two consecutive levels into two consecutive levels}

Let $L_2(k)$ denote the union of the two middle levels of $B_{k}$. 

\begin{obs}\label{obs:L2} Any weak embedding of $L_2(k)$ to $L_2(n)$ is {\em distance-preserving}, i.e., the distance between any two elements of $L_2(k)$ is the same as the distance between their images in $L_2(n)$.
\end{obs} 
\begin{proof} The proof is by induction on $k$.
The statement is trivially true for $k=0,1$.
Take a weak embedding $f: L_2(k)\to L_2(n)$.
Pick an arbitrary element $x\in L_2(k)$ and denote the (unique) element at distance $k$ from it by $\bar x$.
The distance between $x$ and any element other than $\bar x$ is preserved by induction.
Take a neighbor $y$ of $\bar x$ --- 
using induction, the distance of $f(y)$ and $f(x)$ is $k-1$.
This implies that there are exactly $\lfloor k \rfloor$ neighbors of $f(y)$ that fall on a shortest path between $f(x)$ and $f(y)$, thus at distance $k-2$ from $f(x)$.
Since $y$ has exactly $\lfloor k \rfloor$ neighbors in $L_2(k)$ that are different from $\bar x$, and $f$ maps each of them on a shortest path between $x$ and $y$ by induction, there are no more neighbors of $f(y)$ that could be on a shortest path between $x$ and $y$, thus all of them have distance $k$ from $f(x)$.
This implies that $f(\bar x)$ must be at distance $k$ from $f(x)$.
\end{proof}

\begin{cor}\label{weak2strong} Any weak embedding of $L_2(k)$ to $L_2(n)$ is also a strong embedding.
\end{cor}

\begin{cor} For any two elements $p,q \in B_n$ that are on the same level or on consecutive levels at distance $k$, there is a unique\footnote{Unique up to composition with an element of $Aut(L_2(k))$.} embedding of $L_2(k)$ to $B_n$ whose image contains both $p$ and $q$.
\end{cor}
\begin{proof} Since we need to preserve order relations when embedding a poset, it is determined which two levels of $B_n$ we need to embed into, and can thus apply Observation \ref{obs:L2}.
\end{proof}

We will denote the above unique embedding of $L_2(k)$ to $B_n$ by $L_2[p;q]$.
Sometimes we will also use $L_2[p;q]$ to denote the ``upside-down reversal'' of this poset, which is different when $k$ is even, as in addition to the level $L$ of $p$ and $q$, it uses the level above $L$ instead of the level below $L$  --- this will not lead to confusion, as from the context it will always be clear which two levels we embed into.

\subsection{NP-complete 3-uniform hypergraph coloring problems}

We will use the {\bf NP}-completeness of {\sc mon-nae-3-sat}, which is (equivalent to) the problem of deciding whether the vertices of a $3$-uniform hypergraph are properly $2$-colorable, and {\sc 3-rainbow}, which is the problem of deciding whether the vertices of a $3$-uniform hypergraph are $3$-colorable, such that every hyperedge contains each color exactly once (such colorings are called {\em rainbow}).
The {\bf NP}-completeness of {\sc mon-nae-3-sat} was proved by Lov\'asz \cite{Lovasz} (it also follows from Schaefer's dichotomy theorem~\cite{Schaefer}), but we could not find our {\sc 3-rainbow} problem in the literature; it is an easy exercise to show that is {\bf NP}-complete.
For completeness, we sketch a proof independently discovered by Jukka Suomela~\cite{Suomela} and Antoine Amarilli~\cite{Amarilli}.

\begin{proof}(Suomela; Amarilli)
We reduce to our problem whether the chromatic number of a graph is at most $3$.
Construct a $3$-uniform hypergraph \HH from a graph $G$ as follows.
The vertices of \HH are the vertices and edges of $G$, and the edges of \HH are the triples $\{(u,v,uv)\mid uv$ is an edge of $G\}$.
It is straightforward to see that \HH has a rainbow $3$-coloring if and only if $G$ has a proper $3$-coloring.
\end{proof}

\section{Proof of Theorem~\ref{thm:d}}
This section contains the proof of Theorem~\ref{thm:d}.
The problem is trivially contained in {\bf NP}, thus it is enough to prove that it is {\bf NP}-hard to decide whether $P\subset B_{h(P)}$ for an input poset $P$ that has a smallest and a largest element.
The reduction is from {\sc mon-nae-3-sat}, the problem of deciding whether the vertices of a $3$-uniform hypergraph \HH are properly $2$-colorable.\\

The vertices of \HH will be denoted by $v_1,\ldots, v_n$.
The height of the poset $P$ will be $3n$ and it will be a union of two sets, $P_1$ and $P_2$.
The restriction of $P$ to $P_1$ is isomorphic to a subposet of $B_{3n}$, which is over the elements $X=\{a_1, b_1, c_1, a_2, b_2, \ldots, c_n\}$.
First we describe $P_1$, then we will describe the elements of $P_2$ via their relations to $P_1$.
If an element $p\in P$ is mapped to some set $S$, then we write $p\to S$.
The question will be to decide whether $P$ embeds to $B_{3n}$ or not.\\

$P_1$ contains every subset of $X$ with at most $9$ elements, except the pairs of the form $\{a_i,b_i\}$ and $\{a_i,c_i\}$, and except that from the sextuples $P_1$ only contains the ones of the form $\{a_i,b_i,a_j,b_j,a_k,b_k\}$ or $\{a_i,c_i,a_j,c_j,a_k,c_k\}$.
(So $P_1$ contains $2\binom n3$ sextuples.)
$P_1$ also contains a chain of length $3n-8$ for every $9$-element set $S$ starting at $S$ and ending in $X$, guaranteeing that $S$ has to be at least $3n-9$ levels lower than $X$.
(This requires at most $\binom n9(3n-10)$ additional elements.)
Thus, the smallest element of $P_1$ is the empty set, and its largest element will be $X$.
This implies that if $P$ embeds into $B_{3n}$, then all the elements of $P_1$ really must be embedded into the same level as the subset of $X$ that was used to define them.

For notational convenience, after a suitable renaming/permutation of the base set, it can even be achieved that each one element set of $P_1$ is mapped to ``itself,'' e.g., $\{a_i\}\to \{a_i\}$.
This also implies that all elements of $P_1$ are mapped to the set defining them.\\

Now we describe the elements of $P_2$, which will depend on the hypergraph \HH.
These are not defined as a subset of $X$ but by their relations to some of the earlier defined subsets.

$P_2$ contains for each vertex $v_i$ an element denoted by $x_i$ such that $\{a_i\}<x_i<\{a_i,b_i,c_i\}$.
Thus if $P$ embeds into $B_{3n}$, then $x_i\to\{a_i,b_i\}$ or $x_i\to\{a_i,c_i\}$.

Finally, $P_2$ contains for every hyperedge $y_l=\{v_i,v_j,v_k\}$ an element $Z_l$ for which $x_i,x_j,x_k<Z_l<\{a_i,b_i,c_i,a_j,b_j,c_j,a_k,b_k,c_k\}$.
Thus if $P$ embeds into $B_{3n}$, then $Z_l$ has to be the unique sextuple that is above $x_i,x_j,x_k$, so its position is determined by the choice of $x_i,x_j,x_k$.\\

As $\{a_i,b_i,a_j,b_j,a_k,b_k\}$ and $\{a_i,c_i,a_j,c_j,a_k,c_k\}$ must have the respective elements of $P$ mapped to them in any weak embedding of $P$ into $B_{3n}$, we have that $P$ embeds into $B_{3n}$ if and only if there is a choice of the position of the elements $x_i$ such that for no hyperedge $\{v_i,v_j,v_k\}$ we have ($x_i\to\{a_i,b_i\}$ and $x_j\to\{a_j,b_j\}$ and $x_k\to\{a_k,b_k\}$) or ($x_i\to\{a_i,c_i\}$ and $x_j\to\{a_j,c_j\}$ and $x_k\to\{a_k,c_k\}$).
But if $x_i\to\{a_i,b_i\}$ corresponds to coloring $v_i$ red and $x_i\to\{a_i,c_i\}$ corresponds to coloring $v_i$ blue, this is clearly equivalent to whether \HH is $2$-colorable or not.\\

This finishes the proof of Theorem~\ref{thm:d}.

\section{Proof of Theorem~\ref{thm:34}}
This section contains the proof of Theorem~\ref{thm:34}.
The problem is trivially contained in {\bf NP}, thus it is enough to prove that it is {\bf NP}-hard to decide whether a given poset $P$ has a weak embedding to the union of the third and fourth levels of some Boolean lattice.
The reduction is from {\sc 3-rainbow}, which is the problem of deciding whether the vertices of a $3$-uniform hypergraph have a rainbow $3$-coloring, i.e., a $3$-coloring where every hyperedge contains each color exactly once.\\

Now we describe the elements of the two-level poset $P$ that we construct from \HH.
Most elements of $P$ will be defined by subsets of an unspecified base set, with the containment relations preserved.

There is an element $\{a,b,c\}$ that can be thought of as the center of $P$, and will be the (unique) element with the most neighbors among all elements of $P$.
In any embedding $\{a,b,c\}$ will have to go somewhere on the third level, as there are several elements that are bigger than it, thus, with a slight abuse of notation, we can suppose that it goes to $\{a,b,c\}$.

For every vertex $v_i$, add an element $\{a,b,c,x_i\}$ to $P$, and for every hyperedge $y_l$, add an element $\{a,b,c,z_l\}$ to $P$ (where $x_i$ and $z_l$ are different for each vertex and for each hyperedge).
We can again suppose that these elements are mapped to ``themselves''.
The way the elements corresponding to vertices and hyperedges can be distinguished is that each $\{a,b,c,x_i\}$ has only one other neighbor, $COL_i$, which thus can be mapped to either $\{a,b,x_i\}$, $\{a,c,x_i\}$ or $\{b,c,x_i\}$, but each $\{a,b,c,z_l\}$ has three further neighbors, $Z_{l,i}$, $Z_{l,j}$ and $Z_{l,k}$, where $y_l=\{v_i,v_j,v_k\}$.
The three neighbors, $Z_{l,i}$, $Z_{l,j}$ and $Z_{l,k}$, need to be mapped in some permutation to the three neighbors of $\{a,b,c,z_l\}$ that are different from $\{a,b,c\}$, i.e., to $\{a,b,z_l\}$, $\{a,c,z_l\}$ and $\{b,c,z_l\}$.

Finally, for every vertex $v_i\in y_l$, there is an element $X_{i,l}$ that has two neighbors, $COL_i$ and $Z_{l,i}$.
Therefore, $X_{i,l}$ and $Z_{l,i}$ must be mapped either to $\{a,b,x_i,z_l\}$ and $\{a,b,z_l\}$, or to $\{a,c,x_i,z_l\}$ and $\{a,c,z_l\}$, or to $\{b,c,x_i,z_l\}$ and $\{b,c,z_l\}$, depending on $COL_i$.\\

We now have to show that $P$ can be weakly embedded into the union of the third and fourth levels of some $B_n$ if and only if \HH has a rainbow $3$-coloring.
If \HH has a rainbow $3$-coloring, then let the image of $COL_i$ be $\{a,b,x_i\}$ if $v_i$ is colored with the first color, $\{a,c,x_i\}$ if $v_i$ is colored with the second color, and $\{b,c,x_i\}$ if $v_i$ is colored with the third color.
From this the embedding of $X_{i,l}$ and $Z_{l,i}$ follows.
The fact that all three colors appear at each hyperedge $y_l=\{v_i,v_j,v_k\}$ guarantees that the three neighbors of $\{a,b,c,z_l\}$, $Z_{l,i}$, $Z_{l,j}$ and $Z_{l,k}$, will not conflict with each other.
If $P$ has an embedding, then a rainbow $3$-coloring of \HH can be derived in a similar way.\\

This finishes the proof of Theorem~\ref{thm:34}.

\begin{remark}
The above constructed poset $P$ can in fact be embedded into the union of the $\chi$-th and $(\chi+1)$-st levels of some $B_n$ if and only if \HH has a rainbow $\chi$-coloring.
To see this, the above proof needs to be modified only in that $\{a,b,c\}$ has to go to some set with $\chi$ elements, and thus there are $\chi$ choices instead of three for the image of each $COL_x$.
\end{remark}

\section{Proof of Theorem~\ref{thm:e}}
This section contains the proof of Theorem~\ref{thm:e}.
The main idea is similar to the proof of Theorem~\ref{thm:34}, but it is more complicated, and we extensively use Observation~\ref{obs:L2}.
As before, the {\bf NP}-membership is trivial, and we prove ${\bf NP}$-hardness by constructing a poset $P$ from a hypergraph \HH such that \HH has a rainbow $3$-coloring if and only if $e(P)\le 1$, i.e., if $P$ can be embedded into some two consecutive levels of a Boolean lattice.
We will denote the union of ``these'' two levels by $L_2$.
This is a bit of a cheating, since we do not know which two levels of which Boolean lattice $P$ could be embedded into.
One can think of $L_2$ either as the union of two sufficiently large levels, or even as the union of two infinite levels, for which our question could be equivalently formulated.\\

Now we describe the elements of the two-level poset $P$.
Most elements of $P$ will be defined by subsets of an unspecified base set, with the containment relations preserved.

There will be two elements, $\{a,b,c\}$ and $\{p,q,r\}$, which play a central role in the construction.
$P$ will contain all $\binom 63+\binom 64$ elements of $L_2[\{a,b,c\};\{p,q,r\}]$.
Observation~\ref{obs:L2} implies that when we weakly embed $P$ to $L_2$, then the distance of the images of $\{a,b,c\}$ and $\{p,q,r\}$ will be six, thus we can conclude that $a,b,c$, $p,q$ and $r$ must all be different.
We can also suppose that $\{a,b,c\}$ and $\{p,q,r\}$ are, respectively, mapped to some elements $\{a,b,c,W\}$ and $\{p,q,r,W\}$ (which we can consider as ``themselves'') where $W$ contains some additional elements of the base set.

For every hyperedge $y_l$, we add $L_2[\{a,b,c,z_l\};\{p,q,r,z_l\}]$ to $P$ (where $z_l$ is different for each hyperedge).
With another application of Observation~\ref{obs:L2}, we can suppose that these elements are mapped to ``themselves $+$ $W$''.

For every vertex $v_i$, we add two neighboring vertices, $\{a,b,c,x_i\}$ and $COL_i$ to $P$.
We can suppose that $\{a,b,c,x_i\}$ is mapped to $\{a,b,c,x_i,W\}$.
$COL_i$ is ideally mapped to one of $\{a,b,x_i,W\}$, $\{a,c,x_i,W\}$ and $\{b,c,x_i,W\}$; for this, we have to eliminate the possibility of it being mapped to some $\{a,b,c,x_i,W\setminus\{w\}\}$.
This is why we needed all the complications compared to the construction used to prove Theorem~\ref{thm:34}.

Finally, for every vertex $x_i$ that is in the hyperedge $z_l$, we add one more degree two element, $X_{i,l}$, that is connected to $COL_i$ and $Z_{l,i}$.
The element $Z_{l,i}$ will be one of the elements from $L_2[\{a,b,c,z_l\};\{p,q,r,z_l\}]$ that neighbors $\{a,b,c,z_l\}$, i.e., one of $\{a,b,z_l\}$, $\{a,c,z_l\}$ and $\{b,c,z_l\}$.
Using Observation~\ref{obs:L2}, we know that $Z_{l,i}$ has to be embedded as one of $\{a,b,z_l,W\}$, $\{a,c,z_l,W\}$ and $\{b,c,z_l,W\}$.
Therefore, $X_{i,l}$ must be mapped either to $\{a,b,x_i,z_l,W\}$, $\{a,c,x_i,z_l,W\}$ or $\{b,c,x_i,z_l,W\}$, and thus $COL_i$ to $\{a,b,x_i,W\}$, $\{a,c,x_i,W\}$ or $\{b,c,x_i,W\}$.\\

We now have to show that \HH has a rainbow $3$-coloring if and only if $P$ can be weakly embedded into $L_2$.
If \HH has a rainbow $3$-coloring, then let the image of $COL_i$ be $\{a,b,x_i,W\}$ if $v_i$ is colored with the first color, $\{a,c,x_i,W\}$ if $v_i$ is colored with the second color, and $\{b,c,x_i,W\}$ if $v_i$ is colored with the third color.
From this the embedding of $X_{i,l}$ and $Z_{l,i}$ follows.
The fact that all three colors appear at each hyperedge $y_l=\{v_i,v_j,v_k\}$ guarantees that the three neighbors of $\{a,b,c,z_l,W\}$, $Z_{l,i}$, $Z_{l,j}$ and $Z_{l,k}$, will not conflict with each other.
If $P$ has an embedding, then a rainbow $3$-coloring of \HH can be derived in a similar way.\\

This finishes the proof of Theorem~\ref{thm:e}.

\section{Proof of Theorem~\ref{thm:j}}
The vertices of the Johnson graph $J(n,k)$ are the $k$-element subsets of an $n$-element base set, and two vertices are connected if they differ in exactly two elements.
A graph $G$ is an {\em induced Johnson subgraph} if there exists an induced copy of $G$ in $J(n,k)$ for some $n,k$.
These graphs were defined in \cite{NaimiShaw} and later studied in \cite{MalikAli}.
The rest of this section contains a sketch of the proof of Theorem~\ref{thm:j}.
(The details are omitted due to the similarity to the proof of Theorem~\ref{thm:e}.)\\

The problem is trivially in {\bf NP}.
We prove {\bf NP}-hardness by constructing a graph $G$ from any $3$-uniform \HH such that $G$ is an induced Johnson subgraph if and only if \HH has a rainbow $3$-coloring.
We need the following variant of Observation~\ref{obs:L2}, which can be similarly proved by induction.

\begin{obs}\label{obs:j} For any $n,k,n',k'$, any embedding of $J(n,k)$ to $J(n',k')$ is distance-preserving.
\end{obs} 

Denote the vertices of \HH by $v_1,\ldots,v_n$ and its hyperedges by $y_1,\ldots,y_m$.
Now we describe how to construct $G$ from \HH.

$G$ will contain a clique on $n+m$ vertices, $x_1,\ldots,x_n,z_1,\ldots,z_m$ (to be mapped to $\{a,b,c,x_i\}$ and $\{a,b,c,z_l\}$), and another clique on $m$ vertices, $z_1',\ldots,z_m'$ (to be mapped to $\{p,q,r,z_l\}$).

$G$ also contains a disjoint copy of $J(6,3)$ (which is the same as the edge graph of a cube) for each pair $z_l,z_l'$, such that $z_l$ and $z_l'$ are contained in this copy of $J(6,3)$ at distance three from each other.
These embeddings are unique due to Observation \ref{obs:j}.

Finally, $G$ contains a vertex $XZ_{i,l}$ (to be mapped to either $\{a,b,x_i,z_l\}$, $\{a,c,x_i,z_l\}$, or $\{b,c,x_i,z_l\}$, depending on the color of $v_i$) for each $v_i\in y_l$.
$XZ_{i,l}$ is connected to $x_i$, $z_l$, and each other vertex of the form $XZ_{i,l'}$.
(Thus the vertices $(x_i,XZ_{i,l},XZ_{i,l'},\ldots)$ form a clique whose size is one more than the degree of $v_i$ in \HH.)\\

Similarly to the proof of Theorem~\ref{thm:e}, it can be proved that the only possible embedding of $G$ to a Johnson graph is the one described in the construction (with a possible extra $W$ in each set).
The fact that $XZ_{i,l}$ and $XZ_{j,l}$ are not neighbors guarantees that every hyperedge must indeed have all three colors.\\

This finishes the sketch of the proof of Theorem~\ref{thm:j}.


\section{Open problems}
We have seen that determining $d(P)$ and $e(P)$ exactly is hard, but is it possible to efficiently approximate these parameters?
By placing a copy of $P$ above another copy of $P$ (i.e., all elements of one copy are larger than any element of the other copy), we obtain a poset $P+P$ for which $d(P+P)=2d(P)+1$ and $e(P+P)=2e(P)+2$, if $P$ has a smallest and a largest element.
This shows that we cannot hope for an additive constant approximation.

On the other hand, by Mirsky's theorem (the dual of Dilworth's theorem), one can partition any poset $P$ on $n$ elements to $h+1=h(P)+1$ antichains on $n_0,\ldots n_h$ elements where $\sum_{i=0}^h n_i=n$, and embed these antichains one above the other.
For an antichain $A_i$ on $n_i$ elements $d(A_i)\le 1+\log n_i$, thus $d(P)\le \sum_{i=0}^{h} 1+\log n_i\le h+h\log \frac nh$.
(It was proved by Gr\'osz, Methuku and Tompkins \cite{GMT} that almost the same upper bound also holds even for $\pi(P)$.
They have also noted that the upper bound is almost sharp if $n_i\approx n/h$ for all $i$.)
From below we trivially have both $\log n\le d(P)$ and $h\le d(P)$, thus this gives a $2$-approximation for $\log d(P)$.

It would be interesting to close the gap between these bounds.

\subparagraph*{Acknowledgements.}
I would like to thank Bal\'azs Patk\'os for calling my attention to the problem, and thank him, Bal\'azs Keszegh, M\'at\'e Vizer, 
Abhishek Methuku and Joshua Cooper for discussions.
I would also like to thank the anonymous reviewers for several useful suggestions on improving the presentation of the results.

\end{document}